\documentclass[
submission
]{dmtcs-episciences}

\usepackage{graphics,color}
\usepackage{amsfonts}
\usepackage{amsmath}
\usepackage{amssymb}
\usepackage{graphicx}

\newtheorem{thm}{Theorem}[section]
\newtheorem{theorem}{Theorem}[section]
\newtheorem{lemma}[thm]{Lemma}
\newtheorem{cor}[thm]{Corollary}

\newtheorem{prop}[thm]{Proposition}

\newtheorem{defn}[thm]{Definition}
\newtheorem{example}[thm]{Example}
\newtheorem{remark}[thm]{Remark}

\usepackage[round]{natbib}

\author{Yuan Li\affiliationmark{1}
  \and Frank Ingram\affiliationmark{2}
  \and Huaming Zhang\affiliationmark{3}}
\title[Certificate complexity and symmetry of nested canalizing functions]{Certificate complexity and symmetry of nested canalizing functions}
\affiliation{
  Department of Mathematics, Winston-Salem State University,USA\\
  Department of Mathematics, Winston-Salem State University,USA\\
 Computer Science  Department, The University of Alabama in Huntsville, USA}
\received{2020-03-10}
\revised{2020-07-29,2021-02-23,2021-08-17}
\accepted{2021-10-07}
\begin{document}
\publicationdetails{23}{2021}{3}{12}{6191}
\maketitle
\begin{abstract}
  Boolean nested canalizing functions (NCFs) have important applications in molecular regulatory networks, engineering and computer science.  In this paper, we study their certificate complexity. For both Boolean values $b\in\{0,1\}$, we obtain a formula for $b$-certificate complexity and consequently, we develop a direct proof of the certificate complexity formula of an NCF. Symmetry is another interesting property of Boolean functions and we significantly simplify the proofs of some recent theorems about partial symmetry of NCFs. We also describe the algebraic normal form of $s$-symmetric NCFs. We obtain the general formula of the cardinality of the set of $n$-variable $s$-symmetric Boolean NCFs for $s=1,\dots,n$. In particular, we enumerate  the  strongly asymmetric Boolean NCFs. 
\end{abstract}
\maketitle

\section{Introduction}

\label{sec-intro} 
Nested canalizing functions (NCFs) were introduced in \cite{Kau2}.
It was shown in \cite{Abd2} that they are identical to the $\emph {unate}$
$\emph{cascade}$ functions, which have been studied extensively in engineering
and computer science. It was shown in \cite{But} that this class of functions produces 
binary decision diagrams with the shortest average path length.  Recently, canalizing and (partially) NCFs have received a lot of attention \cite{He, Abd2, Cla2, Yua4, Lay,  Yua3, Yua1, Mor, Mur2, Shm}.

In \cite{Coo}, Cook et al. introduced the notion of sensitivity as a combinatorial 
measure for Boolean functions. It was extended by Nisan \cite{Nis0, Nis} to block
 sensitivity. Certificate complexity was first introduced by Nisan in 1989 \cite{Nis0,Nis}.
 
 In \cite{Yua1}, a complete characterization for NCFs was obtained via its unique algebraic normal form,  from which explicit formulas enumerating NCFs and their average sensitivity were derived.

In  Theorem 3.6 \cite{Yua3}, the formula of the sensitivity of
any NCF was obtained based on a characterization of NCFs from Theorem 4.2 \cite{Yua1}. It was shown
that  block sensitivity is the same
as  sensitivity for NCFs. 

In \cite{Mor}, the author proved sensitivity is the same as the certificate complexity for $read$-$once$ functions, a class of functions which include the NCFs, characterized as those that can be written using the logical conjunction, logical disjunction, and negation operations, where each variable appears at most once.   

In this paper, we obtain  formulas of  $b$-certificate complexity of an NCF $f$ for $b=0,1$. We denote them by $C_0(f)$ and $C_1(f)$. As a byproduct, we obtain  a direct proof of the certificate complexity formula which is still the same as the formula of  sensitivity \cite{Yua3}.
 
Symmetric Boolean functions have important applications in coding theory and cryptography. In Section 4, based on  Theorem 4.2 in \cite{Yua1}, we study the properties of symmetric NCFs. We significantly simplify the proofs of some  theorems in \cite{Dan}. We also investigate the relationship between the number of layers of an NCF and its number of symmetry levels. For $1\leq s\leq n$, we obtain an explicit formula of the number of $n$-variable $s$-symmetric Boolean NCFs. When $s=n$, this  number is the cardinality of strongly asymmetric NCFs.  Specifically,  we prove that there are more than  $n!2^{n-1}$ strongly asymmetric NCFs when $n\geq 4$.

\section{Preliminaries}

\label{2} In this section, we introduce the definitions and notations. Let
$\mathbb{F}$ be the field $\mathbb{F}_{2}=\{0,1\}$ and $f$: $\mathbb{F}^n\longrightarrow \mathbb{F}$ be a function. It is well known
\cite{Lid} that $f$ can be expressed as a polynomial, called the algebraic
normal form (ANF):
\[
f(x_{1},\dots,x_{n})=\bigoplus_{\substack {0\leq k_i\leq 1 \\i=1,\dots,n}}a_{k_{1}\cdots k_{n}}{x_{1}}^{k_{1}}\cdots{x_{n}}^{k_{n}},
\]
where each $a_{k_{1}\cdots k_{n}}\in\mathbb{F}$. The symbol $\oplus$ stands for addition modulo 2.

A permutation of $[n]=\{1,\dots,n\}$ is a bijection from $[n]$ to $[n]$. 

\begin{defn}\label{def2.1}(Definition 2.3 in \cite{Abd2}, page 168) Let $f$ be a Boolean function in $n$ variables and $\sigma$ 
a permutation of $\{1,\dots,n\}$. The function $f$ is nested canalizing
 in the variable order
$x_{\sigma(1)},\dots,x_{\sigma(n)}$ with canalizing input values
$a_{1},\dots,a_{n}$ and canalized output values $b_{1},\dots,b_{n}$, if it can be
represented in the form

$$f(x_{1},\dots,x_{n})=\left\{
\begin{array}
[c]{ll}%
b_{1} & x_{\sigma(1)}=a_{1}\\
b_{2} & x_{\sigma(1)}= \overline{ a_{1}}, x_{\sigma(2)}=a_{2}\\
b_{3} & x_{\sigma(1)}= \overline{ a_{1}}, x_{\sigma(2)}= \overline{ a_{2}},
x_{\sigma(3)}=a_{3}\\
\vdots & \\
b_{n} & x_{\sigma(1)}= \overline{ a_{1}}, x_{\sigma(2)}= \overline{ a_{2}%
},\dots,x_{\sigma(n-1)}= \overline{ a_{n-1}}, x_{\sigma(n)}=a_{n}\\
\overline{b_{n}} & x_{\sigma(1)}= \overline{ a_{1}}, x_{\sigma(2)}= \overline{
a_{2}},\dots,x_{\sigma(n-1)}= \overline{ a_{n-1}}, x_{\sigma(n)}=\overline{
a_{n}},
\end{array}
\right.$$
where $\overline{a}=a\oplus 1$. The function f is nested canalizing if it is nested canalizing in some variable order. 
\end{defn}

\begin{theorem}\label{th1}(Theorem 4.2 in \cite{Yua1}, page 28)
 Let $n\geq2$. Then $f(x_{1},\dots,x_{n})$ is nested
canalizing iff it can be uniquely written as
\begin{equation}\label{eq2.1}
f(x_{1},\dots,x_{n})=M_{1}(M_{2}(\cdots(M_{r-1}%
(M_{r}\oplus 1)\oplus 1)\cdots)\oplus 1)\oplus b,
\end{equation}
where  $M_{i}=\prod_{j=1}^{k_{i}}(x_{i_{j}}\oplus a_{i_{j}})$,
$i=1,\dots,r$, $k_{i}\geq1$ for $i=1,\dots,r-1$, $k_{r}\geq2$, $k_{1}%
+ \dots + k_{r}=n$, $a_{i_{j}}\in\mathbb{F}_{2}$, $\{i_{j}\mid j=1,\dots,k_{i},
i=1,\dots,r\}=\{1,\dots,n\}$.
\end{theorem}

Because each NCF can be uniquely written as \eqref{eq2.1} and the number $r$ is
uniquely determined by $f$, we can define the following.

\begin{defn}\label{def2.2} For $i=1,...,r$, each $M_i$ of an NCF $f$ in \eqref{eq2.1} is defined as the i-th layer of $f$,  where $r$ is the  number of layers. The vector
$<\!\!k_1,\dots, k_r\!\!>$ is called the layer structure, where $k_i\geq 1$ for $i=1,...,r-1$, $k_{r}\geq2$, $k_{1}%
+ \dots + k_{r}=n$. Each $k_i$ is the size of $M_i$.
\end{defn}
The $i$-th layer $M_i$ is a product of variables and their negations. Such a product  is called \emph{extended monomial} in \cite{Yua1} or \emph{psedomonomial} in \cite{Cur}.  

Note that we always have $k_r\geq 2$ by Theorem \ref{th1}.
Throughout this paper, all NCFs will be assumed to be on $n$ variables, with layer structure $<\!\!k_1,\dots, k_r\!\!>$.

\section{Certificate Complexity of NCFs}
Let $\mathbf{x}=(x_1,\dots,x_n)\in \mathbb{F}^n$. For any subset $S$ of $[n]$, we form $\mathbf{x}^S$ by negating the bits in $\mathbf{x}$ indexed by elements of $S$. We denote  $\mathbf{x}^{\{i\}}$ by $\mathbf{x}^i$.

\begin{defn}(Definition 2.1 in \cite{Ken}, page 45; Definition 1 in \cite{Rub}, page 297)
The sensitivity of f at $\mathbf{x}$, denoted as $s(f,\mathbf{x})$, is the number of indices $i$ such that $f(\mathbf{x})\neq f(\mathbf{x}^i)$.
 The sensitivity of $f$ is ${s(f)=\max_{\mathbf{x}\in \{0,1\}^n}s(f,\mathbf{x})}$.
\end{defn}

Certificate complexity was first introduced by Nisan \cite{Nis0, Nis}, and was initially called sensitive complexity. In the following, we will slightly modify (actually, simplify) the definition  of certificate, but the definition of certificate complexity will remain the same.

\begin{defn}
Let $f(x_1,\dots, x_n)$ be a Boolean function and $\alpha=(a_1,\dots,a_n)\in \mathbb{F}^n $  a word. If $\{i_1,\dots,i_k\}\subset [n]$ and the restriction $f(x_1,\dots, x_n)|_{x_{i_1}=a_{i_1},\dots,x_{i_k}=a_{i_k}}$ is a constant function, where its constant value is  $f(\alpha)$, then we call the subset $\{i_1,\dots,i_k\}$ a certificate of  $f$ on $\alpha$. 
\end{defn}

\begin{defn}
The certificate complexity $C(f,\alpha)$ of $f$ on $\alpha$ is defined as the smallest cardinality of a certificate of $f$ on $\alpha$. The certificate complexity $C(f)$ of $f$ is defined as $\max\{C(f,y)\mid y\in \mathbb{F}^n\}$. The $b$-certificate complexity $C_b(f)$ of $f$, $b\in \mathbb{F}$, is defined as $\max\{C(f,y)\mid y\in \mathbb{F}^n, f(y)=b\}$. 
\end{defn}
Obviously,  $C(f)=\max\{C_0(f), C_1(f)\}$.
\\

\begin{example}
Let $f(x_1,x_2,x_3)=x_1x_2x_3\oplus x_1x_2\oplus x_3$ and $g(x_1,x_2,x_3)=x_1x_2x_3$. We list the certificate complexity of $f$ on every word in Table 1.

It is easy to check  $C(g,(1,1,1))=3$ and $C(g,\alpha)=1$, where $\alpha\neq (1,1,1)$. Hence, $C(g)=3$.
\end{example}

\begin{table}
\begin{center}
\begin{tabular}{| c   | c    |c  |c  |}
\hline
$\alpha$ & $f(\alpha)$ & $C(f,\alpha)$ &  Minimal certificates
  \\ \hline
 (0,0,0) & 0  & 2  & \{1,3\},\{2,3\} \\
 (0,0,1) & 1   & 1  & \{3\}           \\
 (0,1,0) & 0   & 2  & \{1,3\}         \\
 (0,1,1) & 1   & 1  & \{3\}           \\
 (1,0,0) & 0   & 2  & \{2,3\}         \\
 (1,0,1) & 1   & 1  & \{3\}           \\
 (1,1,0) & 1   & 2  & \{1,2\}         \\
 (1,1,1) & 1   & 1  & \{3\}           \\
\hline
\end{tabular}

  \caption{The certificate complexity  for
	$f(x_1,x_2,x_3)=x_1x_2x_3\oplus x_1x_2\oplus x_3$ is 2.}
  \label{table_ex}
  \end{center}
\end{table}

\begin{lemma}\label{lm3.1}
Let $f(x_{1},\dots,x_{n})$ be a Boolean function, $\sigma$ be a permutation on $[n]$, 
and $\beta=(b_1,\dots,b_n)\in \mathbb{F}^n$. If $g=f(x_{\sigma(1)},\dots,x_{\sigma(n)})$ and $h=f(x_1\oplus b_1,\dots, x_n\oplus b_n)$, then the certificate complexities of $f$, $f\oplus 1$, $g$, and  $h$ are the same.
\end{lemma}
\begin{proof}
Note that $f(x_1,\dots, x_n)|_{x_{i_1}=a_{i_1},\dots,x_{i_k}=a_{i_k}}$ is a constant function if and only if 

$$f(x_{\sigma(1)},\dots, x_{\sigma(n)})|_{x_{\sigma(i_1)}=a_{i_1},\dots,x_{\sigma(i_k)}=a_{i_k}}$$ is a constant function. Hence, $C(f,\alpha)=C(g,\alpha)$ for any $\alpha=(a_1,\dots,a_n)\in \mathbb{F}^n$, and thus $C(f)=C(g)$.

The function $f(x_1,\dots, x_n)|_{x_{i_1}=a_{i_1},\dots,x_{i_k}=a_{i_k}}$ is a constant function if and only if
$$h=f(x_1\oplus b_1,\dots, x_n\oplus b_n)|_{x_{i_1}=a_{i_1}\oplus b_{i_1}  ,\dots,x_{i_k}=a_{i_k}\oplus b_{i_k}}$$ is a constant. Hence, $C(f,\alpha)=C(h,\alpha+\beta)$ for any $\alpha$ and given $\beta$. Thus $C(f)=C(h)$ since $\alpha\longmapsto \alpha \oplus\beta$ is a bijection of $\mathbb{F}^n$.

The function $f$ is  constant if and only if $f\oplus 1$ is  constant, thus $C(f)=C(f\oplus 1)$. Specifically, $C_0(f)=C_1(f\oplus 1)$ and $C_1(f)=C_0(f\oplus 1)$.
\end{proof}

In the following, let
\begin{equation}\label{eq3.1}
f(x_{1},\dots,x_{n})=f_r=M_{1}(M_{2}(\cdots(M_{r-1}%
(M_{r}\oplus 1)\oplus 1)\cdots)\oplus 1)
\end{equation}

be an NCF with $r$ layers with monomials $M_1=x_1\cdots x_{k_1}$, $M_2=x_{k_1+1}\cdots x_{k_1+k_2}$, $\dots$, $M_r=x_{k_1+\cdots+k_{r-1}+1}\cdots x_n$.
\\

With a straightforward calculation, we  rewrite  Equation \eqref {eq3.1} as 
\begin{equation}\label{eq3.2}
f(x_{1},\dots,x_{n})=f_r=M_{1}M_2\cdots M_r\oplus M_{1}M_2\cdots M_{r-1}\oplus\cdots \oplus M_1M_2\oplus M_1.
\end{equation}

\begin{lemma}\label{lm3.2}
If $f(x_{1},\dots,x_{n})=x_1\cdots x_n$, then $C_0(f)=1$ and $C_1(f)=n$. Hence, $C(f)=n$.
\end{lemma}
\begin{proof}
It is clear that $C(f,(1,\dots,1))=n$, $f(1,\dots,1)=1$ and $C(f,\alpha)=1$, $f(\alpha)=0$ with $\alpha\neq (1,\dots,1)$.
\end{proof}
Lemma \ref{lm3.2} provides the certificate complexity of an NCF $f_r$ with $r=1$ layer.
We are ready to prove the following theorem. 

\begin{thm}\label{th3.1}

If $f(x_{1},\dots,x_{n})=f_r=M_{1}(M_{2}(\cdots(M_{r-1}%
(M_{r}\oplus 1)\oplus 1)\cdots)\oplus 1)$

and $M_1=x_1\cdots x_{k_1}$, $M_2=x_{k_1+1}\cdots x_{k_1+k_2}$, $\dots$, $M_r=x_{k_1+\cdots+k_{r-1}+1}\cdots x_n$, $r\geq 2$, then 

\[C_0(f_r)=\left\{
\begin{array}
[c]{ll}%
k_2+k_4+\cdots +k_{r-1}+1, \!\!\!\!\quad 2\nmid r\\
k_2+k_4+\cdots +k_{r},\quad \quad  \quad 2\mid r,\\
\end{array}
\right. \]

\[C_1(f_r)=\left\{
\begin{array}
[c]{ll}%
k_1+k_3+\cdots+k_r, \quad \quad \quad 2\nmid r\\
k_1+k_3+\cdots+k_{r-1}+1, \!\!\!\!\quad 2\mid r,\\
\end{array}
\right. \]

\end{thm}

\begin{proof}
We use induction on $r$ to prove the formula of $C_0(f_r)$, and the proof of $C_1(f_r)$ is similar.

If $r=2$, then $f_r=f_2=M_1M_2+M_1=M_1(M_2\oplus 1)$. We will calculate $C(f_2,\alpha)$ for every $\alpha$ such that $f(\alpha)=0$. Since $f(\alpha)=M_1(M_2\oplus 1)(\alpha)=0$ if and only if $M_1(\alpha)=0$ or $M_1(\alpha)=M_2(\alpha)=1$, we divide all such $\alpha$ into two disjoint groups. In the following, we simply write $M_1(\alpha)=0$ as $M_1=0$, $M_1(\alpha)=1$ as $M_1=1$ and so on.

Group 1: $M_1=0$.

In this case, at least one component of $\alpha$ corresponding to a variable in the first layer must be 0. Obviously, for such $\alpha$, $C(f_2,\alpha)=1$.

Group 2: $M_1=1$ and $M_2=1$.

In this case, there is only one possibility, namely, $\alpha=(1,\dots,1)$. It is easy to check that $C(f_2,(1,\dots,1))=k_2$, the number of  variables in $M_2$.

Take the maximal value, we have $C_0(f_2)=k_2$.
\\
If $r=3$, then $f_3=M_1(M_2(M_3\oplus 1)\oplus1)=0$ $\Longleftrightarrow$ $M_1=0$ or $M_1=M_2=M_3\oplus 1=1$. There are two disjoint groups.

Group A: $M_1=0$.

In this group, the certificate complexity for each word is 1.

Group B: $M_1=1$, $M_2=1$ and $M_3=0$.

 In this group, $\alpha=(\overbrace{1,\dots,1}^{k_1},\overbrace{1,\dots,1}^{k_2},\overbrace{*,\dots,*,0,*,\dots,*}^{k_3})$.
 First of all, if we just assign the values of the variables in $M_1$ and $M_2$ (all of those variables in $\alpha$ are 1s), since $f_3=M_1M_2M_3\oplus M_1M_2\oplus M_1$, the variables in $M_3$  never disappear (which means the function is not constant). So, we must assign one $0$ to its corresponding variable in $M_3$ and reduce $f_3$ to $M_1(M_2\oplus 1)$.
 Obviously, in order  to make $f_3$ zero, it is necessary and sufficient to choose all the components of $\alpha$ corresponding to the variables in $M_2$ to assign. So, in this group, for any $\alpha$, we have $C(f_3,\alpha)=k_2+1$. 

In summary, taking the maximal value, yields $C_0(f_3)=k_2+1$.

Now we assume that the formula of $C_0(f_r)$ is true for any NCF with no more than $r-1$ layers. Let us consider 

\[f(x_{1},\dots,x_{n})=f_r=M_{1}(M_{2}(\cdots(M_{r-1}%
(M_{r}\oplus 1)\oplus 1)\cdots)\oplus 1)\]
\[=M_{1}M_2\cdots M_r\oplus M_{1}M_2\cdots M_{r-1}\oplus\cdots \oplus M_1M_2\oplus M_1.\]

If $g(x_{k_1+k_2+1},\dots,x_n)=M_3\cdots M_r\oplus M_3\cdots M_{r-1}\oplus\cdots\oplus M_3M_4\oplus M_3$, we get
 $f_r=M_1(M_2(g\oplus 1)\oplus 1)=M_1M_2g\oplus M_1M_2\oplus M_1$.
It is clear that $f_r=0$ $\Longleftrightarrow$ $M_1=0$ or $M_1=M_2=g\oplus 1=1$.
Next, we will evaluate $C(f_r,\alpha)$ for all $\alpha\in \mathbb{F}$ with $f(\alpha)=0$. 

Case 1: $M_1=0$.

In this case, the certificate complexity of the word is 1.

Case 2: $M_1=1$, $M_2=1$ and $g=0$.

In this case, $\alpha=(\overbrace{1,\dots,1}^{k_1},\overbrace{1,\dots,1}^{k_2},\alpha')$, where $\alpha'$ is a word with length $n-k_1-k_2$.
Obviously, we have $f_r(\alpha)=0$ if and only if $g(\alpha')=0$.

For a fixed $\alpha'$ (equivalently, a fixed $\alpha$), we try to reduce $f_r=M_1M_2g\oplus M_1M_2\oplus M_1$ to zero by assigning values of $\alpha$ to the variables of $f_r$. Since $M_1M_2$ will never be zero, we must try to reduce $g$ to zero first. Once $g$ is zero, we get $f_r=M_1(M_2\oplus 1)$. Hence, we have $C(f_r,\alpha)=k_2+C(g,\alpha')$, and 
\[\max \{C(f_r,\alpha)\mid \alpha, f_r(\alpha)=0\}=k_2+\max \{C(g,\alpha')\mid\alpha',g(\alpha')=0\}=k_2+C_0(g).\]

Since $g$ is an NCF with $r-2$ layers (the first layer is $M_3$, the second layer is $M_4$ and so on), by the induction hypothesis, we have

\[C_0(g)=\left\{
\begin{array}
[c]{ll}%
k_4+k_6+\cdots +k_{r-1}+1, \!\!\!\!\quad 2\nmid (r-2)\\
k_4+k_6+\cdots +k_{r},\quad \quad  \quad 2\mid (r-2).\\
\end{array}
\right. \]

Hence, $\max \{C(f_r,\alpha)\mid \alpha, f_r(\alpha)=0\}=k_2+C_0(g)$ is

$k_2+\left\{
\begin{array}
[c]{ll}%
k_4+k_6+\cdots +k_{r-1}+1, \!\!\!\!\quad 2\nmid (r-2)\\
k_4+k_6+\cdots +k_{r},\quad \quad  \quad 2\mid (r-2)\\
\end{array}
\right.$
$=\left\{
\begin{array}
[c]{ll}%
k_2+k_4+\cdots +k_{r-1}+1, \!\!\!\!\quad 2\nmid r\\
k_2+k_4+\cdots +k_{r},\quad \quad  \quad 2\mid r.\\
\end{array}
\right.$

For any word in Case 1, the certificate complexity is only 1. In summary, we have 

\[C_0(f_r)=\left\{
\begin{array}
[c]{ll}%
k_2+k_4+\cdots +k_{r-1}+1, \!\!\!\!\quad 2\nmid r\\
k_2+k_4+\cdots +k_{r},\quad \quad  \quad 2\mid r.\\
\end{array}
\right. \]
\end{proof}

Because of Lemma \ref{lm3.1}, we have the following.

\begin{cor}\label{Cor1}
If any NCF is written as the one in Theorem \ref{th1}, then

\[C(f_r)=\left\{
\begin{array}
[c]{ll}%
\max \{k_1+k_3+\cdots+k_r,k_2+k_4+\cdots +k_{r-1}+1\}, 2\nmid r\\
\max \{k_1+k_3+\cdots+k_{r-1}+1, k_2+k_4+\cdots +k_{r}\}, 2\mid r.\\
\end{array}
\right. \]

Hence, the certificate complexity of NCF is uniquely determined by the layer structure $(k_1,\dots,k_r)$.
\end{cor}

The above formula is the same as the sensitivity formula $s(f_r)$ in Theorem 3.6 \cite{Yua3}. 

\section{Symmetric Properties of NCFs}
In 1938, Shannon \cite{Sha} recognized that symmetric functions have efficient switch network implementations. Since then, a lot of research has been done on symmetric or partially symmetric Boolean functions. Symmetry detection is important in logic synthesis, technology mapping, binary decision diagram minimization, and testing \cite{Arn, Das, Ala}. In \cite{Dan}, the authors investigated the symmetric and partial symmetric properties of Boolean NCFs. They also presented an algorithm for testing whether a given partial symmetric function is an NCF. In this section, we  use a formula in \cite{Yua1} to give  simple proofs for several theorems in \cite{Dan}. We  also study the relationship between the number of layers $r$ and the number of symmetry levels $s$ (the function is $s$-symmetric) of NCFs.
Furthermore, we  obtain the formula of the number of $n$-variable $s$-symmetric NCFs. In particular, we obtain the formula of the number  of strongly asymmetric NCFs.  We  start this section by providing some basic definitions and  notations.

It is well known that a permutation can be written as the product of disjoint cycles. A $t$-cycle $(i_1\cdots i_t)$ sends $i_k$ to $i_{k+1}$ for $k=1,\dots,t-1$ and sends $i_t$ to $i_1$.
Namely, $i_1\longmapsto i_2\longmapsto\cdots \longmapsto i_t\longmapsto i_1$. A $2$-cycle is called a transposition. Any permutation can be written as a product of transpositions. For example, $(12\cdots n)=((n-1) n)\cdots (2n)(1n)$, where cycles are read right-to-left, as in function composition. 
\begin{defn}
Let $f$ be a Boolean function and $\sigma=(ij)$  a $2$-cycle.  We say that variable $x_i$ is equivalent to $x_j$ if $f(x_1,\dots, x_n)=f(x_{\sigma (1)},\dots, x_{\sigma (n)})$ (namely, $f(\dots, x_i,\dots ,x_j,\dots)=f(\dots, x_j,\dots, x_i,\dots)$).
We denote this  by $i\!\!\sim_f\!\! j$.
\end{defn}
It is clear that $i\!\!\sim_f\!\! j$ is an equivalence relation over $[n]$. We call $\tilde{i}=\{j\mid j\!\!\sim_f\!\! i\}$ a symmetric class of $f$.   If $[n]/\!\!\sim_f =\{\tilde{i}\mid i\in [n]\}$ and $s=|[n]/\!\!\sim_f\!\!|$ is the cardinality of $[n]/\!\!\sim_f$, we call $f(x_1,\dots, x_n)$ $s$-symmetric. 

The definition of $s$-symmetric in this paper is equivalent to the concept of properly $s$-symmetric in \cite{Dan}.

\begin{example}
Let $f(x_1,x_2,x_3,x_4,x_5,x_6,x_7)=x_1x_2x_3x_4\oplus x_5x_6\oplus x_7$. Then $\tilde{1}=\tilde{2}=\tilde{3}=\tilde{4}=\{1,2,3,4\}$, $\tilde{5}=\tilde{6}=\{5,6\}$, $\tilde{7}=\{7\}$. This function is $3$-symmetric.
\end{example}

\begin{defn}
If there is an index $i$ such that $|\tilde{i}|\geq 2$, i.e., $s=|[n]/ \!\!\sim_f\!\!|\leq n-1$, then we call $f$  partially symmetric.  If $s=1$, we call  $f$ totally symmetric or symmetric. 
\end{defn}
Obviously, a function is not partially symmetric if and only if it is $n$-symmetric.

For applications of $1$-symmetric (totally symmetric) Boolean functions to cryptography, see \cite {Ann} from 2005. More results on (totally) symmetric Boolean functions can be found in \cite{Cai, Fra, Cus1,Cus2, Na, Xia, Mai, Mit, Sav}.

\begin{defn}(\cite {Dan}, page 3)
A Boolean function $f(x_1,\dots, x_n)$ is strongly asymmetric if $f(x_1,\dots, x_n)=f(x_{\sigma (1)},\dots, x_{\sigma (n)})$ implies $\sigma$ is the identity.
\end{defn}

Obviously, if a Boolean function is strongly asymmetric then it is $n$-symmetric. 

Let 
\[f(x_1,x_2,x_3,x_4,x_5,x_6)=x_1x_2\oplus x_2x_3\oplus x_3x_4\oplus x_4x_5\oplus x_5x_1\oplus x_6.\]
 It is easy to check that $f$ is $6$-symmetric (not partially symmetric) but not strongly asymmetric since 

$f(x_1,x_2,x_3,x_4,x_5,x_6)=f(x_{\sigma(1)},x_{\sigma(2)},x_{\sigma(3)},x_{\sigma(4)},x_{\sigma(5)},x_{\sigma(6)})$ for $\sigma=(12345)$.

In the following, we frequently use  Equation \eqref{eq2.1}.
Recall that  $a_{i_j}$ is called the canalizing input of the variable $x_{i_j}$.

\begin{prop}(Theorem 3.1 in \cite{Dan})\label{prop1}
All variables in the same symmetric class of an NCF must be in the same layer and have the same canalizing input.
\end{prop}
\begin{proof}
This follows immediately from the uniqueness of Equation \eqref{eq2.1}.
\end{proof}
\begin{remark}\label{Re1}
 In each layer $M_j$, for $j=1,\dots,r$, there are either one or two symmetric classes. If there are two symmetric classes, then one  has  canalizing input 0, and the other has canalizing input 1.

\end{remark}

\begin{prop}\label{prop2}
Let $n\geq2$ and $<\!\!k_1,\dots, k_r\!\!>$ be the layer structure of an NCF $f$. If $k_j\geq3$ for some $j$, then $f$ is partially symmetric. Moreover, if $f$ is $s$-symmetric, then $\lceil \frac{s}{2}\rceil \leq r\leq \min\{n-1,s\}$. 
\end{prop}

\begin{proof}
If $k_j\geq3$ for some $j$, then at least two variables have the  same canalizing inputs by Remark \ref{Re1}. Hence,  this layer has a symmetric class with at least 2 variables and $f$ is partially symmetric. From Equation \eqref{eq2.1}, the last layer has at least two variables, so $r\leq n-1$. We have $r\leq s$ since all variables from different layers must belong to different symmetric classes. Finally, because each layer contributes at most two symmetric classes, we obtain $s\leq 2r$ which means $\lceil \frac{s}{2}\rceil \leq r$.
\end{proof}

\begin{prop}
Let  $f$ be an $s$-symmetric NCF with $r$ layers. Then $r\leq s\leq \min\{2r,n\}$.
\end{prop}
\begin{proof}
It follows from the proof of the previous property.
\end{proof}

\begin{prop}(Theorem 3.2 in \cite{Dan})
If an NCF contains $r_1$ layers with only one canalizing input, and $r_2$ layers with two distinct canalizing inputs, then it is ($r_1+2r_2$)-symmetric.
\end{prop}
\begin{proof}
This is a straightforward application of the uniqueness of  Equation \eqref{eq2.1}.
\end{proof}

Next, we will provide a new and shorter proof for the following proposition.

\begin{prop}\label{prop11}(Theorem 3.7 in \cite{Dan})
An $n$-variable NCF is strongly asymmetric iff it is $n$-symmetric.
\end{prop}
\begin{proof}
We already know that strong asymmetry implies $n$-symmetry.

If an NCF $f$ is $n$-symmetric, i.e., not partially symmetric, then each layer has one or two variables with different canalizing inputs by Proposition \ref{prop2}. If there is a permutation $\sigma$ such that $f(x_{\sigma (1)},\dots, x_{\sigma (n)})=f(x_1,\dots, x_n)$, then, for any $i$, because of the uniqueness of  Equation \eqref{eq2.1}, we know $x_{\sigma(i)}$ and $x_i$ must be in the same layer of $f(x_1,\dots, x_n)$. If this layer has only one variable, then $\sigma(i)=i$. If this layer has two variables $x_i$ and $x_j$ with $i\neq j$, then this layer must be $M=x_i(x_j\oplus 1)$ or $M=(x_i\oplus 1)x_j$. Without loss of the generality, we assume $M=x_i(x_j\oplus 1)$, if $\sigma(i)=j$, then $\sigma(j)=i$ since $x_{\sigma(i)}$ and $x_i$ must be in the same layer. Because $x_{\sigma(i)}(x_{\sigma(j)}\oplus 1)=x_j(x_i\oplus 1)\neq M$,  which is contrary to the uniqueness of Equation \eqref{eq2.1}. Hence, we still have $\sigma(i)=i$. In summary, we always have $\sigma(i)=i$ for any $i$. Therefore, $\sigma$ is the identity and $f$ is strongly asymmetric.
\end{proof}

Strongly asymmetric NCFs were studied in \cite{Dan}, and in Theorem 3.8, the authors enumerated those that have exactly $n-1$ layers, which is the maximal possible number because $k_r\geq 2$. Though they used this assumption in their proof, they apparently omitted it from the theorem statement. We will state the correct version below, and refer the reader to \cite{Dan} (Theorem 3.8) for the proof.
\begin{thm}\label{th9}
There are $n!2^{n-1}$ strongly asymmetric NCFs on $n$ variables with exactly $n-1$ layers.
\end{thm}

In the remainder of this section, we will enumerate the $s$-symmetric NCFs on $n$ variables. As a corollary, we will derive a formula for the number of strongly asymmetric NCFs.

 Let $N(n,s)$ be the cardinality of the set of $n$-variable $s$-symmetric Boolean NCFs. 

\begin{prop}(Proposition 3.9 in \cite{Dan})
If $n\geq 2$, then $N(n,1)=4$.
\end{prop}

\begin{proof}
Since $f$ is 1-symmetric, i.e., totally symmetric, then there is only one layer, and all canalizing inputs must be the same.
So, $f$ must be one of the following functions: $x_1\cdots x_n$, $x_1\cdots x_n\oplus 1$, $(x_1\oplus 1)\cdots (x_n\oplus 1)$ or $(x_1\oplus 1)\cdots (x_n\oplus 1)\oplus 1$. 
\end{proof}

\begin{thm}\label{main}

For $n\geq 2$, the number of strongly asymmetric NCFs is

\[
N(n,n)=\frac{n!}{\sqrt{2}}((1+\sqrt{2})^{n-1}-(1-\sqrt{2})^{n-1}).
\]

\end{thm}
\begin{proof}
By Theorem \ref{th1}, we have 
\[
f(x_{1},\dots,x_{n})=M_{1}(M_{2}(\cdots(M_{r-1}%
(M_{r}\oplus 1)\oplus 1)\cdots)\oplus 1)\oplus b.
\]

1. It is clear that $b$ has two choices.

2. By Proposition \ref{prop2}, we have $\lceil \frac{n}{2}\rceil \leq r\leq n-1$.

3. For each layer structure $<\!\!k_1,\dots, k_r\!\!>$, since $f$ is strongly asymmetric (not partially symmetric), we have $1\leq k_i\leq 2$ by Proposition \ref{prop2}  and thus $k_r=2$ due to $k_r\geq 2$ always. There are 

\[
\binom{n}{k_1}\binom{n-k_1}{k_2}\binom{n-k_1-k_2}{k_3}\cdots\binom{n-k_1-\cdots- k_{r-1}}{k_r}=\frac{n!}{k_{1}!k_{2}!\cdots k_{r}!}
\]
ways to distribute the $n$ variables to the layers.

4. Each layer $M_j$ is either $x_i\oplus a$ or $(x_k\oplus a)(x_l\oplus a\oplus 1)$. In any case, there are two choices. Hence, totally, there are $2^r$ choices.

Combining the information above, we obtain 

\[
N(n,n)=2\sum_{\substack{\lceil \frac{n}{2}\rceil \leq r\leq n-1}}\sum_{\substack{k_{1}+\cdots+k_{r}=n\\1\leq k_{i}\leq 2,
k_{r}=2}}\frac{n!}{k_{1}!k_{2}!\cdots k_{r}!}2^r.
\]

If $n\geq 3$, then it can be written as

\[
N(n,n)=\sum_{\substack{\lceil \frac{n}{2}\rceil \leq r\leq n-1}}\sum_{\substack{k_{1}+\cdots+k_{r-1}=n-2\\1\leq k_{i}\leq 2,
}}\frac{n!}{k_{1}!k_{2}!\cdots k_{r-1}!}2^r.
\]

Suppose that exactly $j$ elements of the set $\{k_1,\dots, k_{r-1}\}$ are equal to 2. We obtain $2j+r-1-j=n-2$ since $k_1+\cdots+k_{r-1}=n-2$. This implies $j=n-r-1$. Hence,

\[
N(n,n)=\sum_{\substack{\lceil \frac{n}{2}\rceil \leq r\leq n-1}}\binom{r-1}{n-r-1}\frac{n!}{2^{n-r-1}}2^r=2n!\sum_{\substack{\lceil \frac{n}{2}\rceil \leq r\leq n-1}}\binom{r-1}{n-r-1}2^{2r-n}.
\]

Let $k=n-r-1$, and so $r=n-k-1$. It is clear that $\lceil \frac{n}{2}\rceil \leq r\leq n-1\Leftrightarrow 0\leq k\leq \lfloor\frac{n}{2}\rfloor-1$. We have

\[
N(n,n)=2n!\sum_{\substack{0 \leq k\leq \lfloor\frac{n}{2}\rfloor-1}}\binom{n-2-k}{k}2^{n-2-2k}.
\]

Since $\binom{n-2-k}{k}=0$ if $k\geq \lfloor\frac{n}{2}\rfloor$, we have

\[
N(n,n)=2n!\sum_{k=0}^{n-2}\binom{n-2-k}{k}2^{n-2-2k}.
\]
We assumed that $n\geq 3$ in the above proof. A direct calculation shows that the formula is still true for $n=2$. 
\\

Let 
\[
p_n(t)=2^{n-2}t^{n-2}(1+\tfrac{t}{2})^{n-2}+2^{n-3}t^{n-3}(1+\tfrac{t}{2})^{n-3}+\cdots +1=\frac{2^{n-1}t^{n-1}(1+\frac{t}{2})^{n-1}-1}{2t(1+\frac{t}{2})-1}.
\]
\\

A direct computation shows that the sum $\sum_{k=0}^{n-2}\binom{n-2-k}{k}2^{n-2-2k}$ is the coefficient of $t^{n-2}$ in the polynomial $p_n(t)$.
 We rewrite $p_n(t)$ as a sum of two rational expressions:
\[
p_n(t)=t^{n-1}\frac{(2+t)^{n-1}}{t^2+2t-1}+\frac{-1}{t^2+2t-1}\cdot
\]
If we write these two rational expressions as power series, it is clear that the smallest order of the terms in the first rational expression is $n-1$. So, the sum $\sum_{k=0}^{n-2}\binom{n-2-k}{k}2^{n-2-2k}$ is the coefficient of $t^{n-2}$ in the power series of $\frac{-1}{t^2+2t-1}$. We have
\[
\frac{-1}{t^2+2t-1}=\frac{-1}{2\sqrt{2}(-1-\sqrt{2}-t)}+\frac{1}{2\sqrt{2}(-1+\sqrt{2}-t)}\cdot
\]

By the formula of geometric series, we obtain
\[
\frac{-1}{t^2+2t-1}=\frac{1}{2\sqrt{2}}\sum_{k=0}^{\infty}(-(1-\sqrt{2})^{k+1}+(\sqrt{2}+1)^{k+1})t^k.
\]
Therefore, the coefficient of $t^{n-2}$ is $\frac{(\sqrt{2}+1)^{n-1}-(1-\sqrt{2})^{n-1}}{2\sqrt{2}}$. Consequently, we obtain

\[
N(n,n)=\frac{n!}{\sqrt{2}}((1+\sqrt{2})^{n-1}-(1-\sqrt{2})^{n-1}).
\]

\end{proof}

When $n=2,3,4$, we have $N(2,2)=4$ and $N(3,3)=24$ and $N(4,4)=240$.

From the above proof, if $n\geq 4$, then
\[
N(n,n)=2n!\sum_{k=0}^{n-2}\binom{n-2-k}{k}2^{n-2-2k}=2n!(2^{n-2}+(n-3)2^{n-4}+\cdots)>2n!2^{n-2}=n!2^{n-1}.
\]

We have obtained the formulas of  $N(n,1)$ and  $N(n,n)$. In the following, we derive the formula $N(n,s)$ for $n\geq 3$ and $2\leq s\leq n-1$.
\begin{thm}
Let $n\geq 3$ and $2\leq s\leq n-1$. Then $N(n,s)$, the number of $n$-variable $s$-symmetric NCFs, is
\[
2\sum_{\substack{\lceil \frac{s}{2}\rceil \leq r\leq s}}\sum_{\substack{k_{1}+\cdots+k_{r}=n\\k_{i}\geq 1,
k_{r}\geq 2}}\frac{n!}{k_{1}!k_{2}!\cdots k_{r}!}\sum_{\substack{t_{1}+\cdots+t_{r}=s\\1\leq t_{i}\leq \min\{2,k_i\}}}\prod_{\substack{1\leq i\leq r}}((t_i-1)(2^{k_i}-2)+1-(-1)^{t_i}).
\]
\end{thm}

\begin{proof}
By Theorem \ref{th1}, we have 
\[
f(x_{1},\dots,x_{n})=M_{1}(M_{2}(\cdots(M_{r-1}%
(M_{r}\oplus 1)\oplus 1)\cdots)\oplus 1)\oplus b.
\]

1. It is clear that $b$ has two choices.

2. By Proposition \ref{prop2}, we get $\lceil \frac{s}{2}\rceil \leq r\leq s$.

3. For each layer structure $<\!\!k_1,\dots, k_r\!\!>$,  there are 

\[
\frac{n!}{k_{1}!k_{2}!\cdots k_{r}!}
\]
ways to distribute the $n$ variables.

4. Each layer $M_i$ contributes $t_i$ symmetry classes, where $1\leq t_i\leq \min\{2,k_i\}$ and $t_1+\cdots+t_r=s$ since $f$ is $s$-symmetric.

5. For each fixed layer  $M_{i}=\prod_{j=1}^{k_{i}}(x_{i_{j}}\oplus a_{i_{j}})$,  there are $2^{k_i}$ choices for $M_i$. Two of them contribute one symmetric class (all canalizing inputs $a_{i_j}$ are equal) and $2^{k_i}-2$ of them contribute two symmetric classes. Since 
\[(t_i-1)(2^{k_i}-2)+1-(-1)^{t_i}=\left\{
\begin{array}
[c]{ll}%
2, \quad \quad \quad \quad \!\!\!t_i=1\\
2^{k_i}-2,\quad t_i=2,\\
\end{array}
\right. \]

there are  $(t_i-1)(2^{k_i}-2)+1-(-1)^{t_i}$ choices of $M_i$ contributing $t_i$ symmetric classes for $t_i=1,2$.

Combining the information above, we obtain the formula of $N(n,s)$.

\end{proof}

We have

\[
\sum_{j=1}^{n}N(n,j)=2^{n+1}\sum_{r=1}^{n-1}\sum_{\substack{k_{1}+\cdots+k_{r}=n\\ k_{i}\geq 1,  k_r\geq 2
}}\frac{n!}{k_{1}!k_{2}!\cdots k_{r}!}.
\]

The right side is the cardinality of the set of  $n$-variable Boolean NCFs according to \cite{Yua1}.

When $n\geq 2$, it is clear that $N(n,s)\geq 1$.  Consequently,
for any $s$, there exists NCFs which are not $s$-symmetric. In particular, there exists  $n$-variable NCFs that are not $(n-1)$-symmetric (Corollary 3.3 in \cite{Dan}).

From Corollary 4.9 in \cite{Yua1}, the number of NCFs with $r$ layers is 
\begin{equation}\label{eq9}
2^{n+1}\sum_{\substack{k_{1}+\cdots+k_{r}=n\\ k_{i}\geq 1,  k_r\geq 2
}}\frac{n!}{k_{1}!k_{2}!\cdots k_{r}!}.
\end{equation}

When $r$ is the maximal value $n-1$, the above number can be simplified as $n!2^n$.

\section{Conclusion}

In this paper, we obtained the formulas of the $b$-certificate complexity of any NCF for $b=0,1$.  We extended some results from \cite {Dan} on symmetric and partially symmetric NCFs and we studied the relationship between the number of layers and the number of symmetry levels. We derived the formulas of the cardinality of all $n$-variable $s$-symmetric Boolean NCFs. As a special case, we obtained the number of $n$-variable strongly asymmetric Boolean NCFs.

\acknowledgements
\label{sec:ack}
We greatly appreciate the referees for their patience and insightful comments. In particular,
we  are  very  grateful  to  a  referee  for  his/her constructive suggestions
to significantly simplify our original  formula and receive a much better  formula in Theorem 4.13.

\nocite{*}
\bibliographystyle{abbrvnat}
\bibliography{sample-dmtcs}

\begin{thebibliography}{34}
\providecommand{\natexlab}[1]{#1}
\providecommand{\url}[1]{\texttt{#1}}
\expandafter\ifx\csname urlstyle\endcsname\relax
  \providecommand{\doi}[1]{doi: #1}\else
  \providecommand{\doi}{doi: \begingroup \urlstyle{rm}\Url}\fi

\bibitem[Arnold and Harrison(1963)]{Arn}
R.~F. Arnold and M.~A. Harrison.
\newblock \emph{Algebraic properties of symmetric and partially symmetric
  Boolean functions}, ieee transactions of electronic computers, vol: ec-12
  issue: 3, pp. 244-251 edition, 1963.

\bibitem[Butler et~al.(2005)Butler, Sasao, and Matsuura]{But}
J.~T. Butler, T.~Sasao, and M.~Matsuura.
\newblock \emph{Average path length of binary decision diagrams}, ieee
  transactions on computers,54,pp. 1041-1053 edition edition, 2005.

\bibitem[Cai et~al.(1996)Cai, Green, and Thierauf]{Cai}
J.-Y. Cai, F.~Green, and T.~Thierauf.
\newblock \emph{On the correlation of symmetric functions}, math. syst. theory,
  vol. 29, no. 3, pp. 245-258 edition, 1996.

\bibitem[Canteaut and Videau(2005)]{Ann}
A.~Canteaut and M.~Videau.
\newblock \emph{Symmetric Boolean functions}, ieee transactions on information
  theory, vol. 51, no. 8, pp. 2791-2811. edition, 2005.

\bibitem[Castro et~al.(2018)Castro, Gonz$\acute{a}$lez, and Medina]{Fra}
F.~N. Castro, O.~E. Gonz$\acute{a}$lez, and L.~A. Medina.
\newblock \emph{Diophantine equations with binomial coefficients and
  perturbation of symmetric Boolean functions}, ieee transactions on
  information theory, vol 64, no.2, pp.1347-1360 edition, 2018.

\bibitem[Cook et~al.(1986)Cook, Dwork, and Reischuk]{Coo}
S.~A. Cook, C.~Dwork, and R.~Reischuk.
\newblock \emph{Upper and lower time bounds for parallel random access machines
  without simultaneous writes}, siam j. comput, 15, pp. 87-89 edition, 1986.

\bibitem[Curto et~al.(2013)Curto, Veliz-Cuba, and Youngs]{Cur}
C.~Curto, V.~I. Veliz-Cuba, and N.~Youngs.
\newblock \emph{The Neural Ring: An Algebraic Tool for Analyzing the Intrinsic
  Structure of Neural Codes}, bull math biol, 75 (9) edition, 2013.

\bibitem[Cusick and Li(2005)]{Cus1}
T.~W. Cusick and Y.~Li.
\newblock \emph{$k$-th order symmetric SAC Boolean functions and bisecting
  binomial coefficients}, discrete applied mathematics, 149 (2005) 73-86
  edition, 2005.

\bibitem[Cusick et~al.(2008)Cusick, Li, and St$\breve{a}$nic$\breve{a}$]{Cus2}
T.~W. Cusick, Y.~Li, and P.~St$\breve{a}$nic$\breve{a}$.
\newblock \emph{Balanced symmetric functions over $GF(p)$}, ieee transactions
  on information theory, vol 54, pp.1304-1307 edition, 2008.

\bibitem[Das and Sheng(1971)]{Das}
S.~R. Das and C.~L. Sheng.
\newblock \emph{On detecting total or partial symmetry of switching function},
  ieee transactions on computers, vol: c-20 , issue: 3, march 1971, pp. 352-355
  edition, 1971.

\bibitem[He and Macauley(2016)]{He}
Q.~He and M.~Macauley.
\newblock \emph{Stratification and enumeration of Boolean functions by
  canalizing depth}, physica d 314 (2016), pp. 1-8. edition, 2016.

\bibitem[Jarrah et~al.(2007)Jarrah, Raposa, and Laubenbacher]{Abd2}
A.~Jarrah, B.~Raposa, and R.~Laubenbacher.
\newblock \emph{Nested canalyzing, unate cascade, and polynomial}, physica d
  233 (2007), pp. 167-174 edition, 2007.

\bibitem[Kadelka et~al.(2017{\natexlab{a}})Kadelka, Kuipers, and
  Laubenbacher]{Cla2}
C.~Kadelka, J.~Kuipers, and R.~Laubenbacher.
\newblock \emph{The influence of canalization on the robustness of Boolean
  networks}, physica d, 353-354 (2017), 39-47. edition, 2017{\natexlab{a}}.

\bibitem[Kadelka et~al.(2017{\natexlab{b}})Kadelka, Li, Kuipers, Adeyeye, and
  Laubenbacher]{Yua4}
C.~Kadelka, Y.~Li, J.~Kuipers, J.~O. Adeyeye, and R.~Laubenbacher.
\newblock \emph{Multistate nested canalizing functions and their networks},
  theoretical computer science, 675,2 (2017), 1-14 edition, 2017{\natexlab{b}}.

\bibitem[Kauffman et~al.(2003)Kauffman, Peterson, Samuelesson, and
  Troein]{Kau2}
S.~A. Kauffman, C.~Peterson, B.~Samuelesson, and C.~Troein.
\newblock \emph{Random Boolean network models and the yeast transcription
  network}, proc. natl. acad. sci 100 (25) (2003), pp. 14796-14799. edition,
  2003.

\bibitem[Kenyon and Kutin(2004)]{Ken}
C.~Kenyon and S.~Kutin.
\newblock \emph{Sensitivity, block sensitivity, and $l$-block sensitivity of
  Boolean functions}, information and computation, 189 (2004), pp. 43-53
  edition, 2004.

\bibitem[Layne et~al.(2012)Layne, Dimitrova, and Macauley]{Lay}
L.~Layne, E.~Dimitrova, and M.~Macauley.
\newblock \emph{Nested canalyzing depth and network stability}, bull math biol,
  74 (2012), pp. 422-433. edition, 2012.

\bibitem[Li and Qi(2006)]{Na}
N.~Li and W.-F. Qi.
\newblock \emph{Symmetric Boolean functions depending on an odd number of
  variables with maximum algebraic immunity}, ieee transactions on information
  theory, vol 52, no.5, may, pp.2271-2273 edition, 2006.

\bibitem[Li and Adeyeye(2019)]{Yua3}
Y.~Li and J.~O. Adeyeye.
\newblock \emph{Maximal sensitivity of Boolean nested canalizing functions},
  theoretical computer science, 791, (2019), 116-122 edition, 2019.

\bibitem[Li and Xiang(2007)]{Xia}
Y.~Li and Z.-H. Xiang.
\newblock \emph{To determine symmetric $PC(k)$ Boolean functions by its
  definition}, sichuan daxue xuebao, 44 (2007) no.2, 209-212 edition, 2007.

\bibitem[Li et~al.(2013)Li, Adeyeye, Murrugarra, Aguilar, and
  Laubenbacher]{Yua1}
Y.~Li, J.~O. Adeyeye, D.~Murrugarra, B.~Aguilar, and R.~Laubenbacher.
\newblock \emph{Boolean nested canalizing functions: A comprehensive analysis},
  theoretical computer science, 481, (2013), 24-36 edition, 2013.

\bibitem[Lidl and Niederreiter(1977)]{Lid}
R.~Lidl and H.~Niederreiter.
\newblock \emph{Finite Fields}, cambridge university press, new york edition,
  1977.

\bibitem[Maitra and Sarker(2002)]{Mai}
S.~Maitra and P.~Sarker.
\newblock \emph{Maximum nonlinearity of symmetric Boolean functions on odd
  number of variable}, ieee transactions on information theory, vol 48, no. 9,
  pp.2626-2630, edition, 2002.

\bibitem[Mishchenko(2003)]{Ala}
A.~Mishchenko.
\newblock \emph{Fast computation of symmetries in Boolean function}, ieee
  transactions on computer-aided design of integrated circuits and system, vol.
  22, no. 11, november 2003 pp. 1588-1593. edition, 2003.

\bibitem[Mitchell(1990)]{Mit}
C.~J. Mitchell.
\newblock \emph{Enumerating Boolean functions of cryptographic significance},
  j. crypto., vol. 2, no 3, pp. 155-170 edition, 1990.

\bibitem[Moriznmi(2014)]{Mor}
H.~Moriznmi.
\newblock \emph{Sensitivity, block sensitivity, and certificate complexity of
  unate functions and read-once functions}, in: diaz j., lanese i., sangiorgi
  d. (eds) theoretical computer science. tcs 2014. lecture notes in computer
  science, vol 8705. springer, berlin, heidelberg. edition, 2014.

\bibitem[Murrugarra and Laubenbacher(2012)]{Mur2}
D.~Murrugarra and R.~Laubenbacher.
\newblock \emph{The number of multistate nested canalizing functions}, physica
  d: nonlinear phenomena, 241, 929-938 edition, 2012.

\bibitem[Nisan(1989)]{Nis0}
N.~Nisan.
\newblock \emph{CREW PRAMs and decision trees}, proc. 21th acm stoc (1989), pp.
  327-335 edition, 1989.

\bibitem[Nisan(1991)]{Nis}
N.~Nisan.
\newblock \emph{CREW PRAMs and decision trees}, siam j. comput, 20 (6) (1991),
  pp. 999-1007 edition, 1991.

\bibitem[Rosenkrantz et~al.(2019)Rosenkrantz, Marathe, Ravi, and Stearns]{Dan}
D.~J. Rosenkrantz, M.~V. Marathe, S.~S. Ravi, and R.~E. Stearns.
\newblock \emph{Symmetric properties of nested canalyzing functions}, discrete
  mathematics and theoretical computer science, dmtcs vol. 21:4, 2019, $\#$19
  edition, 2019.

\bibitem[Rubinstein(1995)]{Rub}
D.~Rubinstein.
\newblock \emph{Sensitivity VS. block sensitivity of Boolean functions},
  combinatorica 15 (2) (1995), pp. 297-299 edition, 1995.

\bibitem[Savicky(1994)]{Sav}
P.~Savicky.
\newblock \emph{On the bent Boolean functions that are symmetric}, europ. j.
  combin, vol. 15, pp. 407-410 edition, 1994.

\bibitem[Shannon(1938)]{Sha}
C.~Shannon.
\newblock \emph{A symbolic analysis of relay and switching circuits}, aiee
  trans, 57: 713-723 edition, 1938.

\bibitem[Shmulevich and Kauffman(2004)]{Shm}
I.~Shmulevich and S.~A. Kauffman.
\newblock \emph{Activities and sensitivities in Boolean network models},
  physical review letters vol 93, no. 4, 23 july 2004, 048701. edition, 2004.

\end{thebibliography}
\label{sec:biblio}

\end{document}